\documentclass[12pt,a4paper]{amsart}
\usepackage[margin=2.5cm]{geometry}

\usepackage[utf8]{inputenc}
\usepackage{amsmath, amssymb, amsfonts,amsthm,amsopn,amscd}
\usepackage{dsfont}
\usepackage{booktabs}

\usepackage{titletoc}
\usepackage[shortlabels]{enumitem}

\DeclareMathOperator{\rk}{rk}

\def\cA{\mathcal{A}}
\def\cF{\mathcal{F}}

\def\cM{\mathcal{M}}
\def\cI{\mathcal{I}}
\def\cP{\mathcal{P}}
\def\cZ{\mathcal{Z}}
\def\ve{\varepsilon}

\def\R{\mathbb{R}}
\def\C{\mathbb{C}}
\def\N{\mathbb{N}}

\def\id{\text{id}}

\usepackage{upgreek}
\usepackage{xcolor}
\usepackage[breaklinks]{hyperref}
\hypersetup{
    colorlinks,
    linkcolor={red!40!black},
    citecolor={blue!50!black},
    urlcolor={blue!20!black}
}

\DeclareMathOperator{\tr}{tr}

\DeclareMathOperator{\SEP}{SEP}

\DeclareMathOperator{\h}{h}
\DeclareMathOperator{\Wg}{Wg}

\newcommand{\ket}[1]{\mathinner{|#1\rangle}}
\newcommand{\braket}[2]{\mathinner{\langle #1|#2\rangle}}
\newcommand{\dyad}[1]{| #1\rangle \langle #1|}
\newcommand{\ot}[0]{\otimes}
\newcommand{\one}[0]{\mathds{1}}
\renewcommand{\a}{\alpha}

\newcommand{\vv}{\ket{v_1}, \dots, \ket{v_n}}
\newcommand{\vvv}{\dyad{v_1} \ot \cdots \ot \dyad{v_n}}
\newcommand{\cdn}{(\C^d)^{\otimes n}}

\newcommand{\nn}{\nonumber}

\definecolor{alizarin}{rgb}{0.82, 0.1, 0.26}

\newtheorem{theorem}{Theorem}
\newtheorem{proposition}[theorem]{Proposition}
\newtheorem{lemma}[theorem]{Lemma}
\newtheorem{corollary}[theorem]{Corollary}

\newtheorem{remark}[theorem]{Remark}

\newtheorem{thmA}{Theorem}

\newcommand{\overbar}[1]{\mkern 1.5mu\overline{\mkern-1.5mu#1\mkern-1.5mu}\mkern 1.5mu}

\makeatletter
\@namedef{subjclassname@2020}{%
	\textup{2020} Mathematics Subject Classification}
\makeatother

\begin{document}

\setcounter{tocdepth}{3}
\contentsmargin{2.55em} 
\dottedcontents{section}[3.8em]{}{2.3em}{.4pc} 
\dottedcontents{subsection}[6.1em]{}{3.2em}{.4pc}
\dottedcontents{subsubsection}[8.4em]{}{4.1em}{.4pc}

\title[
Dimension-free entanglement detection
]{
Dimension-free entanglement detection \\ in multipartite Werner states
}

\date{\today}

\author{Felix Huber}
\address{
Institute of Theoretical Physics, 
Jagiellonian University, 
30-348 Krak\'{o}w, 
Poland}
\email{felix.huber@uj.edu.pl}
\thanks{FH was supported by the FNP through TEAM-NET (POIR.04.04.00-00-17C1/18-00).}

\author{Igor Klep}
\address{
Faculty of Mathematics and Physics, 
University of Ljubljana, 
Slovenia}
\email{igor.klep@fmf.uni-lj.si}
\thanks{IK was supported by the Slovenian Research Agency grants J1-2453, J1-8132, N1-0217 and P1-0222.}

\author{Victor Magron}
\address{
LAAS-CNRS \& Institute of Mathematics from Toulouse, 
France}
\email{victor.magron@laas.fr}
\thanks{VM was supported by the French Research Agency grants  ANR-18-ERC2-0004-01 and ANR-19-PI3A-0004, the EU's Horizon 2020 research and innovation programme 813211, and the PHC Proteus grant 46195TA}

\author{Jurij Vol{\v c}i{\v c}}
\address{
Department of Mathematical Sciences, 
University of Copenhagen, 
Denmark}
\email{jv@math.ku.dk}
\thanks{JV was supported by the National Science Foundation grant DMS-1954709.}

\begin{abstract}
Werner states are multipartite quantum states that are 
invariant under the diagonal conjugate action of the unitary group.
This paper gives a complete characterization of their entanglement that 
is independent of the underlying local Hilbert space: for every entangled Werner state there exists a
dimension-free entanglement witness.
The construction of such a witness is formulated as an optimization problem.
To solve it, two semidefinite programming hierarchies are introduced. 
The first one is derived using real algebraic geometry applied to positive polynomials in the entries of a Gram matrix, and
is complete in the sense that for every entangled
Werner state it converges to a witness.
The second one is based on a sum-of-squares certificate for the positivity of trace polynomials in noncommuting variables, and is a relaxation
that involves smaller semidefinite constraints.
\end{abstract}



\maketitle

\tableofcontents

\section{Introduction}

\subsection{Entanglement}
An $n$-partite {quantum state} with local dimension~$d$ is represented by a 
positive semidefinite matrix with trace one
in the space $L(\cdn)$ of linear operators acting on $\cdn$.
A quantum state $\varrho \in L(\cdn)$ is said to be 
{\bf separable} or classically correlated, 
if it can be written as a convex combination of {product states}
\begin{equation*}
\sum_{i} p_i \varrho_i^{(1)} \otimes \dots \otimes \varrho_i^{(n)}\,,
\end{equation*}
where $\varrho_i^{(j)} \in L(\C^d)$ are states,
and $p_i\geq 0$ satisfy $\sum_i p_i = 1$. 
We denote the set of separable states on $n$ systems with $d$ levels each as $\SEP(d,n)$.
A state is termed {\bf entangled} if it is not separable~\cite{GUHNE20091}.
The detection of entanglement can be done with linear operators known as {\bf entanglement witnesses}.
These are operators $\mathcal{W}\in L(\cdn)$ for which 
$\tr(\mathcal{W} \varrho) \geq 0$ 
holds for 
all separable states~$\varrho$
and 
$
\tr(\mathcal{W} \varphi) < 0$
holds for at least one entangled state $\varphi$.
Note that since separable sets are defined as the convex hull of product states,
it suffices to ascertain that $\tr(\mathcal{W} \varrho) \geq 0$ holds for all product states~$\varrho$ only.

Nevertheless, 
characterizing the set of entangled states is computationally hard \cite{gurvits03} and it helps to restrict the set of states under consideration. 
Here we focus on {\bf Werner states}~\cite{werner89,PhysRevA.63.042111,CKMR07, MaassenKuemmerer2019,huber2020positive}:
these are invariant under
the diagonal action of the unitary group $\mathcal{U}_d$, i.e.,
$\varrho = U^{\otimes n} \varrho (U^\dag)^{\otimes n}$ for all $U \in \mathcal{U}_d$.
As a consequence of the Schur-Weyl duality \cite[Theorem 9.3.1]{Procesi2007LieGroups}, Werner states are linear combinations of permutation operators.
Note that an element $\sigma$ in the symmetric group $S_n$ acts on the Hilbert space $(\C^d)^{\otimes n}$ by permuting its tensor factors.
With some abuse of notation we can then write a Werner state $\varrho$ as
\begin{equation}\label{eq:introeta}
\varrho=\sum_{\sigma \in S_n} r_\sigma \sigma,\qquad r_\sigma  \in \C\,.
\end{equation}
That is, Werner states are parametrized by elements of the group algebra $\C S_n$.
It is interesting to note that Werner states have applications both 
in quantum information theory as well as in many-body physics:
they were introduced to show that entanglement and 
non-locality are distinct concepts~\cite{werner89},
and their entanglement structure can be used to characterize 
correlations close to phase transitions 
in magnetic systems~\cite{PhysRevA.89.032330}.

To detect entanglement in Werner states, it is 
easy to see that one can restrict to entanglement witnesses $\mathcal{W}$ 
that exhibit the same invariance as the states.
Thus we can represent them by
$  w = \sum_{\sigma \in S_n} w_\sigma \sigma$ with $w_\sigma \in \C$.
We say that $w\in\C S_n$ is a {\bf dimension-free witness} if the operator $\mathcal{W}$ represented by $w$ is a witness regardless of the local dimension $d$.

The description \eqref{eq:introeta} of Werner states removes the underlying local Hilbert space, which is especially useful when the latter has large dimension. This raises a natural question: can the entanglement of Werner states be also described in a dimension-independent manner? Furthermore, 
does such a dimension-free description yield a computationally efficient procedure 
for entanglement detection?
This paper provides affirmative answers to both questions.

For three-partite Werner states, a description of entanglement without referring to the local dimension
was given in \cite{PhysRevA.63.042111}.
Here we present a complete characterization for the entire class of Werner states 
(for any number of local systems). 
To efficiently detect their entanglement, we employ
semidefinite programming hierarchies.

\subsection{SDP hierarchies}
Semidefinite programming (SDP) hierarchies have emerged as powerful tools applicable 
to a wide range of problems in quantum information theory~\cite{
PhysRevA.74.052306,
Cavalcanti_2016,
mironowicz2018_phdthesis,
Wang2018_phdthesis,
BertaBorderiFawziScholz2021}. 
Solving an SDP \cite{anjos2011handbook} means 
minimizing a linear function under linear matrix inequality constraints, which is a convex problem.
The advantages of SDPs lie with the existence of efficient algorithms,  the ready availability of numerical solvers, 
and ability to provide solution certificates~\cite{Vandenberghe1996, Blekherman2013}.
When formulated in this framework, many quantities that are otherwise difficult to compute can be approximated 
by a converging sequence of increasingly larger SDP instances. 

A well-known example is the Navascu\'es-Pironio-Ac\'in hierarchy for finding the maximum violation levels of Bell inequalities~\cite{Navascu_s_2008}. 
This hierarchy gives a sequence of outer approximations to the set of correlations that 
can be obtained from quantum systems of arbitrarily large (even infinite-dimensional) local Hilbert space.
This is in contrast with the hierarchies used in entanglement detection: 
here the available hierarchies detect entanglement of quantum states where the local dimension is {\em fixed}~\cite{
PhysRevA.69.022308,
PhysRevA.70.062309,
PhysRevLett.106.190502,
PhysRevA.80.052306,
Lancien_2015,
PhysRevA.70.062317,
PhysRevA.96.032312,
HarrowNatarajanWu2017,
acin2021prept}. 
While extremely powerful for small systems, these hierarchies are afflicted by the exponential scaling of the problem size with the local Hilbert space dimension.

It is thus of interest to not only approach non-locality, 
but also entanglement in a dimension-free manner. 
With the help of methods from commutative and noncommutative polynomial optimization \cite{Las01sos,SchererHol2006,klep2020optimization},
we use our dimension-free characterization of Werner states to
detect their entanglement with SDP hierarchies that do not 
depend on the local Hilbert space dimension.

\subsection{Main results}

The first main contribution of this paper reveals the dimension-independent nature of entanglement for Werner states.

\begin{thmA}\label{thm:dim-free1}
	For all $d,n$ and 
	every entangled Werner state $\varrho \in L(\cdn)$ there exists a 
	dimension-free witness $w \in \C S_n$ detecting it.
\end{thmA}

For the proof of Theorem \ref{thm:dim-free1} see Corollary \ref{cor:dim-free} below.
Thus the set of separable Werner states can be described using hyperplanes of the form
$w = \sum_{\sigma \in S_n} w_\sigma \sigma$ whose $n!$ parameters are entirely independent of the local dimension. A key step in bypassing the dependence on the local dimension  is replacing the usual description \eqref{eq:introeta} of Werner states in terms of the symmetric group with a special weighted version arising from the representation theory of $S_n$.
A characterization of entangled Werner states without referring to the local Hilbert space is given in Theorem \ref{thm:werner-gram}.

The second main contribution of this paper are two SDP hierarchies for finding dimension-free entanglement witnesses for Werner states as in Theorem \ref{thm:dim-free1}. 
Both of them arise from the optimization problem for a given Werner state $\varrho$:
\begin{equation}\label{eq:common}  
\begin{aligned}
\ve^* = &\underset{\ve\in\R,\,w \in\C S_n}{\inf}
&& \ve  \\
& \text{subject to}
&& \tr(\mathcal{W}\varrho)=-1\,, \\
& && \mathcal{W}\quad \text{is represented by }w\,, \\
& && w+\ve \quad \text{is a dimension-free witness}\,.
\end{aligned}
\end{equation}
Then $\varrho$ is entangled if and only if $\ve^*<1$.
The difference between our two hierarchies stems from encoding the last constraint in \eqref{eq:common}.

The first hierarchy~\ref{eq:sdp_f_o_m} encodes positivity of $w+\ve$ on product states with polynomials in {commuting} variables $z_{ij}$ that represent angles between unit vectors.
These variables can be seen as entries of a positive semidefinite Gram matrix with~$1$s on the diagonal,
corresponding to extremal points of the set of separable states.
Using Putinar's Positivstellensatz from real algebraic geometry,
optimization of a polynomial in variables $z_{ij}$ over all Gram matrices with $1$s on the diagonal can then be cast as a sequence of SDPs as in Lasserre's hierarchy \cite{Las01sos}.

\begin{thmA}\label{thm:fsth}
Let $\varrho$ be a Werner state. Then $\varrho$ is entangled if and only if a term in the hierarchy \ref{eq:sdp_f_o_m} returns a value less than 1, in which case it also produces a dimension-free entanglement witness for $\varrho$. 
\end{thmA}

The second hierarchy \ref{eq:sdp_tracepoly}
applies the trace polynomial optimization framework
introduced
by the second, third and fourth author~\cite{klep2020optimization}
to the correspondence between positive trace polynomials and Werner state entanglement witnesses by the first author~\cite{huber2020positive}.
Trace polynomials are 
polynomial-like expressions in noncommuting variables $x_1,\dots,x_n$ and traces of their products. 
It turns out that positivity of a trace polynomial over all tracial von Neumann algebras can be characterized with a sum-of-squares certificate \cite[Theorem 4.4]{klep2020optimization}. Since matrices are special cases of tracial von Neumann algebras, we can use sum-of-squares representations of trace polynomials to confirm their positivity on matrices. 
Finally, since Werner state witnesses correspond to trace polynomials positive on tuples of positive semidefinite matrices (\cite[Theorem 16]{huber2020positive}, 
also see Theorem \ref{thm:trace_poly_formulation}),
this leads to the hierarchy \ref{eq:sdp_tracepoly} for entanglement detection.\looseness=-1

\begin{thmA}\label{thm:sndh}
Let $\varrho$ be a Werner state. If a term in the hierarchy \ref{eq:sdp_tracepoly} returns a value less than 1, then $\varrho$ is entangled and a corresponding dimension-free entanglement witness is produced. 
\end{thmA}

While the hierarchy \ref{eq:sdp_f_o_m} is complete
since it converges to an entanglement witness for every entangled Werner state, it is not clear whether 
\ref{eq:sdp_tracepoly} detects every entangled Werner state. However, the latter hierarchy's first steps involve much smaller semidefinite constraints than the hierarchy ~\ref{eq:sdp_f_o_m},
which makes it more suitable for concrete calculations. 
As a demonstration, we use \ref{eq:sdp_tracepoly} to produce an exact entanglement witness for a 4-partite Werner state, 
for which the Peres-Horodecki criterion (i.e., a negative partial transpose signals entanglement~\cite{PhysRevLett.77.1413, Horodecki1996})  
fails (Section \ref{sec:exa}).\looseness=-1

\section{Dimension-free entanglement witnesses for Werner states}

In this section we present a parametrization of Werner states with the group algebra of the symmetric group that admits a dimension-free characterization of entanglement. 
Our approach generalizes \cite{PhysRevA.63.042111} where tripartite Werner states were considered.
We start by introducing notions from the representation theory of the symmetric group that are required throughout the paper. Then we build towards Theorem \ref{thm:werner-gram} which relates entanglement of Werner states with a certain system of polynomial inequalities that is independent of the local dimension. As a consequence we prove the existence of dimension-free entanglement witnesses (Corollary \ref{cor:dim-free}).
 
The group algebra $\C S_n$ has a canonical conjugate-linear involution $\dag$ given by inverting group elements, $(\sum_{\sigma\in S_n} a_\sigma\sigma)^\dag=\sum_{\sigma\in S_n} \overline{a_\sigma}\sigma^{-1}$.
Furthermore, there is 
a natural trace
$$\tau:\C S_n\to\C,\qquad \tau(a)=n!a_{\id}$$
where $a_{\id}$ is the coefficient of the identity $\id$ in $a\in\C S_n$. 
Throughout the paper we view $\C S_n$ as a Hilbert space with the scalar product induced by $\tau$; that is, $\frac{1}{\sqrt{n!}}S_n$ is an orthonormal basis of $\C S_n$.
We define the set of {\bf states} as
\begin{equation*}
    \{r \in \C S_n :\, r = a a^\dag,\, a \in \C S_n,\, \tau(r) =1\}
\end{equation*}
The terminology is justified by Lemma \ref{lem:CSn2Hilbert}(2) below.
Note that $r = aa^\dag$ for some $a\in\C S_n$ if and only if $r$ is a positive semidefinite element of the
finite-dimensional $C^*$-algebra $\C S_n$, which is further equivalent to $\Phi(r)\succeq0$ for every $*$-representation $\Phi$ of $\C S_n$.

\bigskip
We now outline the necessary facts from the representation theory of the symmetric group~\cite{fultonharris,Procesi2007LieGroups}.
To each partition $\lambda\vdash n$ is associated an irreducible representation of $S_n$ 
(cf. \cite[Chapter 4]{fultonharris}); let $\chi_\lambda$ be its character. 
Let $\{\omega_\lambda : \lambda  \vdash n\}$ be a complete set of centrally primitive idempotents for $\C S_n$ \cite[Section 3.4]{fultonharris}.
They can be written as
\begin{equation*}
	\omega_\lambda = \frac{\chi_\lambda(\id)}{n!} \sum_{\sigma \in S_n}  \chi_\lambda(\sigma) \sigma^{-1}\,,
\end{equation*}
where $\chi_\lambda(\id)$ is both the multiplicity and the dimension of 
the irreducible representation corresponding to $\lambda$ in $\C S_n$.

The trace $\tau$ can be seen as the linear extension of 
the character of the regular representation of $S_n$. If $\sigma\in S_n$, then the Schur column orthogonality relations \cite[Section 2.2]{fultonharris} imply
\begin{equation}\label{eq:char_reg}
    \tau(\sigma) 
    = \sum_{\lambda \vdash n} \chi_\lambda(\id)\chi_\lambda(\sigma) 
    = 
    \begin{cases}
      \sum_{\lambda  \vdash n} \chi_\lambda(\id)^2 = n! &\text{if } \sigma=\id\,, \\
       0 & \text{otherwise}\,.
    \end{cases}
\end{equation}
Here $\chi_\lambda(\id)$ is both the multiplicity and the dimension of 
an irreducible representation corresponding to $\lambda$ in $\C S_n$. In particular, $\tau(\omega_\lambda) = \chi_\lambda^2(\id)$.

\bigskip

Let $\eta_d$ be the representation of $S_n$ on $(\C^d)^{\otimes n}$ that permutes the tensor factors,
$$\eta_d(\sigma)(\ket{v_1}\otimes\cdots\otimes \ket{v_n})=\ket{v_{\sigma^{-1}(1)}}\otimes\cdots\otimes \ket{v_{\sigma^{-1}(n)}}$$
for $\sigma\in S_n$ and $\vv\in \C^d$. 
Under $\eta_d$, the idempotents $\omega_\lambda$ are mapped to 
the central Young projections $p_\lambda = \eta_d(\omega_\lambda)$. These satisfy
\begin{align*}
 p_\lambda^2 &= p_\lambda = p_\lambda^\dag          \,,\nn\\
 p_\lambda p_\mu &= p_\lambda \delta_{\lambda \mu}  \,,\nn\\      
 \eta_d(\sigma) p_\lambda &= p_\lambda \eta_d(\sigma) \quad \forall \sigma \in S_n\,.
 \end{align*}
Importantly, they form a resolution of the identity
\begin{equation*}
 \sum_{\substack{\lambda \vdash n \\ 
                  \h(\lambda) \leq d}
        }
    p_\lambda = \one \quad \in L(\C^d)\,.
\end{equation*}

By \cite[Proposition 9.3.1]{Procesi2007LieGroups},
\begin{equation}\label{eq:kernel}
\ker \eta_d = \sum_{
	\substack{\lambda \vdash n \\ 
		\h(\lambda) > d}
} \omega_\lambda \cdot \C S_n \,.
\end{equation}
Let
\begin{equation*}
	 J_d = \sum_{
	        \substack{\lambda \vdash n \\ 
	          \h(\lambda) \leq d}
	          } \omega_\lambda \cdot \C S_n \,.
	\end{equation*}
Then $J_d$ and $\ker\eta_d$ are complementary (both as orthogonal subspaces and ideals) in $\C S_n$.
Furthermore, $J_1\subset J_2\subset \cdots \subset J_n=J_{n+1}=\cdots=\C S_n$.
Next consider the map $\mu_d:\C S_n\to L((\C^d)^{\otimes n})$ defined as
\begin{equation*}
    \mu_d(r) = n! \Wg(d,n) \eta_d(r)\,,
\end{equation*}
where 
\begin{align}\label{eq:def_Wg}
\Wg(d,n) &
 = \frac{1}{ n!} \sum_{
        \substack{\lambda \vdash n \\ 
          \h(\lambda) \leq d}
          }
        \frac{\tau(\omega_\lambda)}{\tr(p_\lambda)} 
          p_\lambda 
\end{align}
is the {\bf (Formanek-) Weingarten operator}~\cite{CollinsSniady2006, procesi2020note}. 
The action of $\Wg(d,n)$ scales each isotypic component according to its multiplicity in $\C S_n$ and $L(\cdn)$. 
Note that the restriction of $\mu_d$ to $J_d$ is bijective onto the image of $\eta_d$ since $J_d=(\ker\eta_d)^\perp$.

\smallskip
The definition of $\mu_d$ is motivated by the following properties:
 
\begin{lemma}\
\label{lem:CSn2Hilbert}
\begin{enumerate}[\rm (1)]
 \item
 For all $a\in J_d$ and $b \in \C S_n$ it holds that 
 \begin{equation*}
 \tr(\mu_d(a) \eta_d(b)) = \tau(a b)\,.
 \end{equation*}

 \item Let $r \in J_d$. Then $r$ is a state  if and only if 
 $\mu_d(r)$ is a state in $L(\cdn)$. 
 
 \end{enumerate}
\end{lemma}

\begin{proof}
(1) Since $J_d$ is an ideal, we have $ab\in J_d$. Next,
\begin{equation}\label{eq:twotr}
\tau(\omega_\lambda)\cdot\tr(\eta_d(\omega_\lambda c)) 
= \tr(p_\lambda)\chi_\lambda(\id)\cdot \chi_\lambda(c)
\end{equation}
for all $c\in\C S_n$ and $\lambda\vdash n$. Indeed, both sides of \eqref{eq:twotr} restrict to traces on the central simple algebra $\omega_\lambda\cdot \C S_n$. 
As $\eta_d(\omega_\lambda)=p_\lambda$ and $\tau(\omega_\lambda)=\chi_\lambda(\id)^2$, 
\eqref{eq:twotr} holds for $c=\id$.
Since traces of central simple algebras over $\C$ are unique up to a scalar multiple, 
we thus conclude that \eqref{eq:twotr} holds for every $c\in \C S_n$.
Therefore
\begin{align}\label{eq:inner_CS_tensor}
\begin{split}
  \tr(\mu_d(a) \eta_d(b)) \ 
  &= \tr\Big( n! \Wg(d,n) \eta_d(a) \eta_d(b)\Big) \\
    &= \sum_{
        \substack{\lambda \vdash n \\ 
          \h(\lambda) \leq d}
          }
        \frac{\tau(\omega_\lambda)} {\tr(p_\lambda)} 
          \tr(\eta_d(\omega_\lambda a b))  \\
    &= \sum_{
        \substack{\lambda \vdash n \\
          \h(\lambda) \leq d}
          }
        \chi(\id) \chi_\lambda(ab)
    = \tau(ab)\,,
   \end{split}
\end{align}
by \eqref{eq:twotr} and \eqref{eq:char_reg}.

(2)  $(\Rightarrow)$  Suppose $\tau(r)=1$ and $r=aa^\dag$ for some $a\in \C S_n$. Then
$$\mu_d(r)= n! \Wg(d,n) \eta_d(aa^\dag) = n! \Wg(d,n)^{1/2} \eta_d(a)\eta_d(a^\dag) \Wg(d,n)^{1/2} \succeq 0$$
and
$\tr(\mu_d(r)) = \tau(r) = 1$
by \eqref{eq:inner_CS_tensor}, so $\mu_d(r)$ is a state in $L(\cdn)$.

\phantom{(2)} $(\Leftarrow)$ Suppose that $\mu_d(r)$ is a state in $L(\cdn)$. Then $\mu_d(r)\succeq0$ implies $p_\lambda \eta_d(r)\succeq0$ for all $\lambda\vdash n$ with $h(\lambda)\leq d$. Therefore $\eta_d(r)\succeq0$ because $r\in J_d$. Since the restriction of $\eta_d$ to $J_d$ is a $*$-embedding, we have $r=aa^\dag$ for some $a\in J_d$. Finally, $\tau(r)=\tr(\mu_d(r))=1$ by \eqref{eq:inner_CS_tensor}.
\end{proof}

Let $z = (z_{ij} \,:\, 1\leq i<j\leq n)$ be a tuple of $\binom{n}{2}$ complex variables.
Denote by $Z$ the $n \times n$ matrix over $\C[z, \overbar{z}]$
with entries $Z_{ii} = 1$, $Z_{ij} = z_{ij}$ and $Z_{ji} = \overbar{z_{ij}}$ for $i<j$.
Let 
\begin{equation*}
\cZ = \{\a \in \C^{\binom{n}{2}} \,:\; Z(\a)\geq 0\}
\end{equation*} 
be the corresponding bounded spectrahedron, 
also known as the {elliptope}~\cite{Vinzant2014}.
For $d\in\N$ also let
$$\cZ_d = \{\a \in \cZ \,:\; \rk Z(\a)\le d\}.$$
Note that $\cZ_1\subset \cZ_2\subset \cdots \subset \cZ_n=\cZ_{n+1}=\cdots= \cZ$.
Furthermore, $\alpha\in\cZ_d$ if and only if $\alpha_{ij}=\braket{v_i}{v_j}$ for some unit vectors $\vv\in\C^d$.
 
To each $w=\sum_{\sigma\in S_n}w_{\sigma}\sigma \in \C S_n$ we assign the polynomial
\begin{equation}\label{eq:genmatfun}
f_w = \sum_{\sigma\in S_n}w_\sigma\prod_{i=1}^n z_{i\sigma(i)} \in \C[z,\overbar{z}]\,
\end{equation}
where $z_{ii}$ denotes $1$ and $z_{ji}$ for $i<j$ denotes $\overbar{z_{ij}}$.
These polynomials are also known as generalized matrix functions \cite{marcusminc}.
If 
$\alpha\in\cZ$ is given as $\alpha_{ij}=\braket{v_i}{v_j}$ for
unit vectors $\vv \in \C^d$, then 
\begin{equation}
\label{eq:pro}
f_w(\alpha) 
= \sum_{\sigma\in S_n}w_\sigma\prod_{i=1}^n \braket{v_i}{v_{\sigma(i)}}
=\tr\big(\eta_d(w) (\vvv)\big)
\end{equation}
by \cite[Theorem 9.6.1]{Procesi2007LieGroups}.

We require two technical lemmas.

\begin{lemma}\label{l:rk1}
Let
$$u_d = \sum_{\substack{\lambda \vdash n \\ \h(\lambda) > d}} 
\omega_\lambda \in \ker \eta_d\,.$$
Then $f_{u_d}$ is nonnegative on $\cZ$ and $\cZ_d=\cZ\cap\{f_{u_d}=0\}$.
\end{lemma}

\begin{proof}
Let $\cM$ be the set of all $(d+1)$-minors of $Z$, 
and let $\cP$ be the set of all principal $(d+1)$-minors of $Z$.
Observe that $\cP\subseteq \{f_w\colon w\in \C S_n\}$, 
and $\alpha\in\cZ_d$ if and only if $p(\alpha)=0$ for all $p\in\cP$.
On the other hand, if $\cI$ is the ideal in $\C[z,\overline{z}]$ generated by $\cM$, then 
$\{f_w\colon w\in\ker\eta_d \}=\cI\cap \{f_w\colon w\in\C S_n \}$ 
by \cite[Section 11.6.1]{Procesi2007LieGroups}.
Therefore $\alpha\in\cZ_d$ if and only if $f_w(\alpha)=0$ for all $w\in\ker \eta_d$. 

Let $\vv\in\C^n$ be arbitrary unit vectors, and denote $V=\vvv \in L(\cdn)$.
Since $V$ and $\eta_n(\omega_\lambda)$ are projections, we have
$$ \tr(\eta_n(\omega_\lambda)V)=0 \implies \tr(\eta_n(a\omega_\lambda)V)=0$$
for every $a\in \C S_n$ by the Cauchy-Schwarz inequality.
Furthermore, since the projections $\eta_n(\omega_\lambda)$ with $\h(\lambda)>d$ are orthogonal and generate $\ker\eta_d$ as a left ideal, we have $\tr(\eta_n(u_d)V)\ge0$ and
\begin{align*}
&\tr(\eta_n(w)V)=0 \quad \forall w\in \ker\eta_d \\ 
\iff & \tr(\eta_n(\omega_\lambda)V)=0 \quad \forall \h(\lambda)>d \\
\iff & \tr(\eta_n(u_d)V)=0 \,.
\end{align*}

Finally, since every $\alpha\in \cZ$ is of the form $\alpha_{ij}=\braket{v_i}{v_j}$ for some unit vectors $\vv\in\C^n$, the preceding two paragraphs and \eqref{eq:pro} imply that 
$f_{u_d}$ is nonnegative on $\cZ$,
and $\alpha\in\cZ_d$ if and only if $f_{u_d}(\alpha)=0$.
\end{proof}

\begin{lemma}\label{l:rk2}
Suppose that $p\in\C[z,\overline{z}]$ is nonnegative on $\cZ_d$, and let $\ve>0$. Then there is $u=u^\dag\in\ker \eta_d$ such that $p+\ve+f_u$ is nonnegative on $\cZ$.
\end{lemma}

\begin{proof}
By Lemma \ref{l:rk1} we have $f_{u_d}(\alpha)>0$ for every $\alpha\in \cZ\setminus\cZ_d$. 
Since $p+\ve$ is positive on $\cZ_d$, it is also positive on some Euclidean open subset $U\subset\cZ$ that contains $\cZ_d$.
Since $\cZ\setminus U$ is compact, there exists $M>0$ such that
$$M\cdot \min_{\alpha\in \cZ\setminus U} f_{u_d}(\alpha)\ge -\min_{\alpha\in \cZ\setminus U}(p(\alpha)+\ve).$$
Then $Mu_d\in\ker\eta_d$ and $p+\ve+ f_{Mu_d}=p+\ve+M f_{u_d}$ is nonnegative on $\cZ$.
\end{proof}

We are now ready to treat entanglement of Werner states in a dimension-independent manner.

\begin{theorem}\label{thm:werner-gram}

Given a state $r\in J_d$, the following are equivalent:
\begin{enumerate}[\rm (i)]
\item $\mu_d(r)$ is entangled;

\item there is $w=w^\dag\in\C S_n$ such that
 \begin{align*}
  f_{w}(\alpha) &\geq 0 \quad \forall \alpha\in \cZ\,,\\
  \tau(w r) &< 0.
 \end{align*}
\end{enumerate}
 \end{theorem}
 
 \begin{proof}
 (ii)$\Rightarrow$(i) By Lemma \ref{lem:CSn2Hilbert}(1) we have $\tr(\eta_d(w)\mu_d(r))=\tau(wr)<0$, and by \eqref{eq:pro} we have
 $$\tr\big(\eta_d(w)(\vvv)\big)\ge0$$
 for all unit vectors $\vv\in\C^N$, and $N\in\N$. Since every separable state is a conic combination of operators of the form $\vvv$, we conclude that $\tr\big(\eta_N(w)\varrho\big)\ge0$ for all $\varrho\in\SEP(N,n)$ and $N\in\N$. In particular, $\eta_d(w)$ is an entanglement witness for $\mu_d(r)$.
 
 (i)$\Rightarrow$(ii) Since $\mu_d(r)$ is entangled, there exists $w_0=w_0^\dag\in \C S_n$ such that $\eta_d(w_0)$ is an entanglement witness for $\mu_d(r)$. Therefore $\tau(w_0r)=\tr(\eta_d(w_0)\mu_d(r))<0$ and $f_{w_0}$ is nonnegative $\cZ_d$. 
 Let $\ve=-\frac12\tau(wr)>0$.
 By Lemma \ref{l:rk2} 
 there exists $u\in\ker\eta_d$ such that
 $$f_{w_0}+\ve+f_u=f_{w_0+\ve\id+u}$$
 is nonnegative on $\cZ$. Thus $w=w_0+\ve\id+u$ satisfies $\tau(wr)=\frac12\tau(w_0r)<0$ and $f_w(\alpha)\ge0$ for all $\alpha\in\cZ$.
\end{proof}

\begin{corollary}\label{c:indep}
Let $r\in J_d$ and $d<e$. Then:
\begin{enumerate}[\rm (1)]
	\item $\mu_d(r)$ is a state if and only if $\mu_{e}(r)$ is a state;
	\item $\mu_d(r)$ is entangled if and only if $\mu_{e}(r)$ is entangled.
\end{enumerate}
\end{corollary}

\begin{proof}
By definition we have $J_d\subseteq J_{e}$. Then (1) holds by Lemma \ref{lem:CSn2Hilbert} and (2) holds by Theorem \ref{thm:werner-gram}.
\end{proof}

\begin{remark}\rm
The assumption $r\in J_d$ in Corollary \ref{c:indep} is necessary; if $d<n$ then there exists $s\in  J_{e}\setminus J_d$, and so $\mu_d(ss^\dag)=0\succeq0$ and $\mu_{e}(ss^\dag)\not\succeq0$.
Furthermore, the direct analog of Corollary \ref{c:indep} fails for $\eta_d$ 
(which is a more conventional parametrization of Werner states than $\mu_d$),
as already the maximally mixed state fails to remain normalized.
Actually, 
the inadequacy of using $\eta_d$ for studying entanglement in a dimension-free way
stretches beyond normalization. For example, if $r=\id-\frac12 (12) \in \C S_2$, then $\frac{1}{\tr(\eta_2(r))}\eta_2(r)$ is a separable state and  $\frac{1}{\tr(\eta_3(r))}\eta_3(r)$ is an entangled state \cite{werner89}.
\end{remark}

An witness $w=w^\dag\in\C S_n$ is called {\bf dimension-free} if
$\tr(\eta_d(w)\varrho)\ge0$ for all $\varrho \in\SEP(d,n)$ and all $d\in\N$.
Another important consequence of Theorem \ref{thm:werner-gram} is the existence of dimension-free witnesses.

\begin{corollary}\label{cor:dim-free}
For all $d,n$ and 
every entangled Werner state $\varrho \in L(\cdn)$ there exists a 
dimension-free witness $w \in \C S_n$ detecting it.
\end{corollary}

\begin{proof}
If a state $\varrho =\mu_d(r)$ is entangled, then $w$ from Theorem \ref{thm:werner-gram}(ii) is a dimension-free entanglement witness for $\varrho$, which follows from the proof of (ii)$\Rightarrow$(i).
\end{proof}

\begin{remark}\label{rem:para}\rm
Theorem \ref{thm:werner-gram} shows that describing Werner states in $L(\cdn)$ with $J_d$ via $\mu_d$ 
reveal the dimension-free nature of entanglement.
While the map $\mu_d$ is defined using the Weingarten operator and is of a rather representation-theoretic nature, 
its unique preimages in $J_d$ can be computed in a very elementary way if one has access to the more common map $\eta_d$. 
Suppose $A\in L(\cdn)$ is invariant under the diagonal conjugate action of $\mathcal{U}_d$. 
There is a unique $a=\sum_{\pi\in S_n} a_\pi \pi \in J_d$ 
such that $A=\mu_d(a)$. By Lemma \ref{lem:CSn2Hilbert}(1), the coefficients of $a$ are given by
$$a_\sigma 
= \frac{1}{n!}\tau(a\sigma^{-1}) 
= \frac{1}{n!}\tr\big(\mu_d(a)\eta_d(\sigma^{-1})\big)
=\frac{1}{n!}\tr\big(A\eta_d(\sigma)^\dag\big)$$
for $\sigma\in S_n$.

Alternatively, if say $\varrho = \eta_d(ss^\dag)/\tr(\eta_d(ss^\dag))$ with $s \in \C S_n$ is given, 
then $r$ in $\varrho = \mu_d(r)$  is proportional to
\begin{equation*}
    \widetilde{\Wg}(d,n)^{-1} \Big( \sum_{
	\substack{\lambda \vdash n \\ 
	\h(\lambda) \leq d}
	} \omega_\lambda  \Big) ss^\dag
\end{equation*}
with an overall normalization such that the coefficient of $\id$ is $1/n!$,
and where $\widetilde{\Wg}(d,n)^{-1}$ is the inverse of the analog of $\Wg(d,n)$ in $\C S_n$,
\begin{equation*}
 \widetilde{\Wg}(d,n)^{-1} 
 = n! \sum_{
        \substack{\lambda \vdash n \\ 
          \h(\lambda) \leq d}
          }
        \frac{\tr(p_\lambda)} {\tau(\omega_\lambda)} 
          \omega_\lambda  \,.
\end{equation*}
\end{remark}

\section{Entanglement witnesses via commutative polynomial optimization}
With the help of Theorem~\ref{thm:werner-gram}
we now show how  semidefinite programming allows us to find entanglement witnesses for Werner states.
The key idea is that
finding entanglement witnesses of this type can be formulated 
as optimizing a multilinear polynomial over a compact semialgebraic set.
We recall the matrix version of Putinar's Positivstellensatz \cite{putinar} 
from real algebraic geometry
in a form suitable for our application.

\begin{corollary}[{Complex version of the matrix Positivstellensatz~\cite[Corollary~1]{SchererHol2006}}]
\label{cor:compmatpss}
A polynomial $q\in \C[z,\overbar{z}]$ is nonnegative on $\mathcal{Z}$
if and only if $q + \varepsilon \in Q$ for every $\varepsilon > 0$,
where
\begin{equation*}
 Q = \Big\{\sum_j p_j^\dag Z p_j\,:\, p_j \in \C[z,\overbar{z}]^n \Big\} \subset \C[z,\overbar{z}]
\end{equation*}
is the quadratic module generated by $Z$.
\end{corollary}
 
Sandwiching $Z$ with polynomials of at most degree $\ell$ yields
the $\ell$-truncated quadratic module
\begin{equation}\label{eq:Ql}
 Q_\ell = \big\{\tr( (u_\ell \ot \one_n)^\dag G (u_\ell \ot \one_n) Z) \,:\, G\succeq 0 \big\}\,, 
\end{equation} 
where $u_\ell$ is the vector of ordered monomials in $z, \overbar{z}$ of degree at most $\ell$, and $G$ is a
$m_\ell \times m_\ell$ matrix 
with $m_\ell = n \binom{n(n-1) + l}{n(n-1)}$.
Clearly, $Q = \bigcup_\ell Q_\ell$. 
Note that $f_w$ can be of degree $n$; to consider whether $f_w +\ve\in Q_\ell$ for some $\ve>0$, it is therefore sensible to restrict $\ell \geq \lceil \frac{n}{2} \rceil$.

A matrix polynomial $P(z) \in \C[z,\overbar{z}]^{n\times n}$ is a sum of squares (SOS) if there is a matrix polynomial 
$S(z) \in \C[z,\overbar{z}]^{m \times n}$ such that $P(z) = S^\dag(z) S(z)$.
By writing $G = Y^\dag Y$, the polynomial matrix $(u_\ell \ot \one_n)^\dag G (u_\ell \ot \one_n)$ 
is easily seen to be SOS, 
\begin{align*}
 (u_\ell \ot \one_n)^\dag G (u_\ell \ot \one_n) &= (Y(u_\ell \ot \one_n))^\dag Y(u_\ell \ot \one_n)
 =\left( \sum_i Y_i (u_\ell)_i \right)^\dag \left(\sum_i Y_i (u_\ell)_i\right)\,,
\end{align*} 
where $Y = (Y_1, \dots, Y_{m_\ell})$ is understood as a block $1\times \frac{m_\ell}{n}$ matrix with $m_\ell \times n$ blocks $Y_i$.

\smallskip

Given $r \in J_d$,  consider the following commutative polynomial optimization problem:
\begin{equation}\label{eq:pop}\tag{POP}
\begin{aligned}
\ve^*= &\inf_{
 \varepsilon \in \R,\ 
 w \in \C S_n
 } &&  \varepsilon \\
 &\text{subject to} && w=w^\dag \\
 &                  && \tau(r w) = -1\\
 &                  && f_w + \varepsilon \geq 0 \text{ on } \mathcal{Z}  \,.\\
\end{aligned}
\end{equation} 
This gives rise to the following hierarchy of SDP relaxations for \ref{eq:pop}, indexed by $\ell\in\N$:
\begin{equation}\label{eq:sdp_f_o_m}\tag{SDP-POP}
\begin{aligned}
\ve_\ell^*= &\inf_{
 \substack{ 
 \varepsilon \in \R,\ 
 w \in \C S_n, \\ 
 G \in L(\C^{m_\ell})
 }
 } &&  \varepsilon \\
 &\text{subject to} && w=w^\dag \\
                    &&& G \succeq0\\
 &                  && \tau(r w) = -1\\
 &                  && f_w + \varepsilon = \tr( (u_\ell^\dag \ot \one_n) G (u_\ell \ot \one_n) Z)\,.\\
\end{aligned}
\end{equation} 

\begin{corollary}
Let $r\in J_d$. Then $\mu_d(r)$ is entangled if and only if $\ve^*_\ell<1$ for some $\ell\in\N$.
\end{corollary}
\begin{proof}
$(\Rightarrow)$ If $\mu_d(r)$ is entangled, then there is $w=w^\dag\in \C S_n$ such that $\tau(rw)<0$ and $f_w|_{\mathcal Z}\ge0$ by Theorem~\ref{thm:werner-gram}. After rescaling $w$ we can assume that $\tau(rw)=-1$. 
By Corollary~\ref{cor:compmatpss}, 
there exists $\ell\in\N$ such that $f_w+\frac12\in Q_\ell$. 
Then $\ve_\ell^*\le \frac12<1$.

$(\Leftarrow)$ Suppose 
$\ve_\ell^*<1$ for some $\ell\in\N$.
Then
$$\tau(r(w+\ve_\ell^* \id))=\tau(rw)+\ve_\ell^*\tau(r)=-1+\ve_\ell^*<0$$
and
$f_{w+\ve_\ell^*\id}$ is nonnegative on $\mathcal{Z}$. Therefore $\mu_d(r)$ is entangled by Theorem~\ref{thm:werner-gram}.
\end{proof}

\begin{remark}\label{rem:size1}\rm
Fix $n\in\N$. The $\ell$th SDP \ref{eq:sdp_f_o_m} has
$$1+n!+n^2\binom{n(n-1)+\ell}{n(n-1)}^2 = O(\ell^{2n(n-1)})$$
real variables ($\ve$, coefficients of $w=w^\dag$, and entries of $G$), and its semidefinite constraint has size $n\binom{n(n-1)+\ell}{n(n-1)}$. 
Thus the size of \ref{eq:sdp_f_o_m} grows polynomially in $\ell$.
\end{remark}

\section{Entanglement witnesses via trace polynomial optimization}

In this section we associate Werner state witnesses with multilinear trace polynomials with certain positivity properties (Theorem \ref{thm:trace_poly_formulation}). 
Thus we translate the problem of finding Werner state witnesses to trace polynomial optimization, and produce a second SDP hierarchy for entanglement detection.

\subsection{Trace polynomials}

Trace polynomials are polynomials in noncommuting variables where some terms are traced, for example
\begin{equation*}
\tr(x_1 x_2) x_3 - \tr(x_2 x_3 x_1)\one
+ 2 \tr(x_1x_3)^2 x_2 + x_1x_3-x_3x_1+\one
\,.
\end{equation*}
Here we only work with linear combinations of terms of the form
\begin{equation*}
T_\sigma 
= \tr(x_{\a_1} \cdots x_{\alpha_r})  \cdots 
\tr(x_{\zeta_1} \cdots x_{\zeta_t}) \,,
\end{equation*}
where $\sigma = (\a_1 \dots \alpha_r) \dots (\zeta_1 \dots \zeta_t)$ is a permutation.
For example, $T_{(132)(4)} = \tr(x_1x_3x_2) \tr(x_4)$. 
As before,
let $\eta_d$ be the representation of $S_n$ on $(\C^d)^{\otimes n}$ that permutes the tensor factors.
Then
a direct calculation in $L(\cdn)$ shows~\cite[Lemma 4.9]{Kostant2009}
\begin{equation}\label{eq:T_sig-eq-tensor}
\tr\left( \eta_d(\sigma) (X_1 \otimes \cdots \otimes X_n)\right)
= T_{\sigma^{-1}}(X_1, \dots, X_n)
\end{equation}
for all $X_1,\dots,X_n \in L(\C^d)$.
In particular,
\begin{equation}\label{eq:perm-trace}
\tr(\eta_d(\sigma)) = d^{N_{\text{cyc}} (\sigma)}
\end{equation}
where $N_{\text{cyc}}(\sigma)$ is the number of cycles in $\sigma$.
This leads to the following consequence of \cite[Theorem 16]{huber2020positive}.

\begin{theorem} \label{thm:trace_poly_formulation}
	Let $\varphi = \sum_{\pi \in S_n} a_\pi \eta_d(\pi)$ be a state, and let $\mathcal{W} = \sum_{\sigma \in S_n} w_\sigma \eta_d(\sigma)$. The following are equivalent:
	\begin{enumerate}[\rm (i)]
		\item $\mathcal{W}$ detects entanglement in $\varphi$;
		\item the trace polynomial $\sum_{\sigma \in S_n} w_\sigma T_{\sigma^{-1}}(x_1, \dots, x_n)$
		satisfies
		\begin{align*}
		&\sum_{\sigma \in S_n} w_\sigma T_{\sigma^{-1}}(X_1, \dots, X_n) \geq 0 
		\quad \forall X_i \in L(\C^d),\ X_i\geq 0\,, \nonumber\\
		&\sum_{\sigma, \pi \in S_n} w_{\sigma} a_\pi d^{N_{\text{cyc}} (\sigma\pi)} < 0\,.
		\end{align*}
	\end{enumerate}
\end{theorem}

\begin{proof}
	The set of separable states $\SEP(d,n)$ is convex and 
	it suffices to ascertain that $\tr(\mathcal{W} \varrho ) \geq 0$ holds for all product states~$\varrho$. 
	With Eq.~\eqref{eq:T_sig-eq-tensor} one has
	\begin{equation*}
	\tr(\mathcal{W} \varrho_1 \otimes \cdots \otimes \varrho_n) 
	= \sum_{\sigma \in S_n} w_\sigma T_{\sigma^{-1}}(\varrho_1, \dots, \varrho_n) \,.
	\end{equation*}
	The expression is multilinear so we can replace the $\varrho_i$ by arbitrary $X_i \geq 0$ in $L(\C^d)$.
	With Eq.~\eqref{eq:perm-trace} it is immediate that 
	\begin{equation*}
	\tr(\mathcal{W} \varphi) 
	= \sum_{\sigma, \pi \in S_n} w_\sigma a_\pi d^{N_{\text{cyc}} (\sigma\pi)} \,. \qedhere
	\end{equation*}
\end{proof}

\subsection{Trace polynomial optimization}

In this subsection we give an alternative way of confirming Werner state entanglement 
using a recently introduced framework for trace polynomial optimization \cite{klep2020optimization}.
The key idea is the following: for the trace polynomials appearing in Theorem~\ref{thm:trace_poly_formulation}, 
instead of requiring positivity in matrix variables of size $d$,
one asks for positivity in operator variables from any tracial von Neumann algebra.
This is of course a stronger requirement; 
however, positivity of trace polynomials over all tracial von Neumann algebras can be exactly described by sums of squares and their traces.

Let $\cM$ be the monoid generated by $x_1,\dots,x_n$ subject to relations $x_j^2=x_j$ for $j=1,\dots,n$. 
Namely, $\mathcal{M}$ is the set of words in $x_1,\dots,x_n$ without consecutive repetitions of letters, and for $v,w\in\cM$ define $vw$ as the concatenation of $v$ and $w$ with consecutive repetitions of letters removed. 
The empty word in $\cM$ is denoted by $1$.
Also define a natural involution $\dag$ that reverses words,  
and an equivalence relation: 
$v\sim w$ if $w$ can be obtained by a cyclic rotation of the letters in $v$.

Denote the equivalence class of $u\in\cM\setminus\{1\}$ by $\uptau(u)$.
The defining relations for $\cM$ (namely $x_j^2 = x_j$ for $j = 1,\dots, n$) describe projections,
and so $\uptau$ simulates a tracial state on a product of projections.
Let $A$ be the complex polynomial ring in symbols $\uptau(u)$ for $u\in\cM\setminus\{1\}$, 
and let $\cA=A\otimes \C\cM$.
Thus $\cA$ is a noncommutative algebra which inherits the involution $*$ from $\cM$. 
Assigning elements from $\cM$ to their equivalence classes $A$-linearly extends to a unital trace map $\uptau:\cA\to A$.
For example, if 
$$
a        = 3i \uptau(x_1 ) x_2 x_1x_3x_2+\uptau(x_2)x_2 \in \cA
$$
then
\begin{align*}
a^\dag       &= -3i \uptau(x_1 ) x_2x_3x_1 x_2+\uptau(x_2)x_2\,, \\
\uptau(a) &= 3i\uptau(x_1 ) \uptau(x_2x_1 x_3)+\uptau(x_2)^2\,.
\end{align*}

Let $a\in A$. Given a von Neumann algebra $\cF$ with a tracial state $\omega:\cF\to\C$ and a tuple $X=(X_1,\dots,X_n)$ of projections $X_j\in\cF$, there is a naturally defined evaluation $a(X)\in\C$, 
determined by $\uptau(x_{j_1}\cdots x_{j_\ell})(X_1,\dots,X_n) = \omega(X_{j_1}\cdots X_{j_\ell})$.

The elements from $\cA$ of the form $\uptau(u_1)\cdots\uptau(u_m)u_0$ for $u_0,\dots,u_m\in \cM$ are called {\em tracial words}. 
Let us fix some total ordering of tracial words that respects their word length. For $\ell\in\N$ let $W_\ell$ be the vector of ordered tracial words in $\cA$ of length at most $\ell$.
Given $a\in A$ let
\begin{equation}\label{eq:KMVcor57}
 \epsilon_\ell=\inf\big\{\epsilon\colon a+\epsilon = \uptau(W_\ell^\dag GW_\ell),\ G\succeq0\big\}\,. 
\end{equation} 
Note that $\uptau(W_\ell^\dag GW_\ell)$ yields a trace of sum of squares in $\cA$.
The value $\epsilon_\ell$ relates to optimization over all tracial von Neumann algebras in the  following way.
\begin{corollary}[{Complex analog of \cite[Corollary 5.7]{klep2020optimization}}]
\label{cor:KMVcor57}
The sequence $(\epsilon_\ell)_\ell$ in Eq.~\eqref{eq:KMVcor57} is decreasing and bounded; let $\epsilon^*$ be its limit.
Then $-\epsilon^*$ is the infimum of $a(X)$ over all tuples $X$ of projections from tracial von Neumann algebras.
\end{corollary}

We now look at the tracial words arising from elements in $S_n$.
Given a permutation $\sigma = (\a_1 \dots \alpha_r) \dots (\zeta_1 \dots \zeta_t) \in S_n$
define
\begin{equation}\label{eq:scaled}
\mathfrak{t}_\sigma =
n^{N_{\text{cyc}}(\sigma)}\uptau(x_{\a_1} \cdots x_{\alpha_r})  \cdots 
\uptau(x_{\zeta_1} \cdots x_{\zeta_t}) \in A \,.
\end{equation}
We extend this notation linearly to the group algebra $\C S_n$. The definition \eqref{eq:scaled} is motivated by the following observation. Let $w\in \C S_n$ and let $X\in L(\C^n)^n$ be a tuple of projections. On one hand, we can evaluate the trace polynomial $T_w$ on $X$ to obtain $T_w(X)\in\C$. On the other hand, $L(\C^n)$ is a tracial von Neumann algebra with the unique tracial state $\frac{1}{n}\tr$; since elements of $A$ can be evaluated at tuples of projections from von Neumann algebras, we can also talk about $\mathfrak{t}_w(X)\in\C$. The choice of the cycle-counting scalar factor in \eqref{eq:scaled} ensures that
\begin{equation}\label{eq:mat2vna}
T_w(X)=\mathfrak{t}_w(X)\,.
\end{equation}
Note that \eqref{eq:mat2vna} is valid only for projections on $\C^n$, and not for those on spaces of other dimensions.

\begin{proposition}\label{prop:oneway}
	Let $r \in J_d$ be a state.
	Suppose that there is a $w=w^\dag\in \C S_n$ such that 
	\begin{align}\label{eq:entchar2} 
	\tau(rw) & =-1\,,\nonumber\\
	\qquad \mathfrak{t}_w+\vartheta &= \uptau(W_\ell^\dag GW_\ell)\,,
	\end{align}
	for some $\vartheta<1$, $\ell\in\N$, and $G\succeq0$. 
	Then $\mu_{e}(r)$ is entangled for every $e\ge d$, with a dimension-free witness $\widetilde{w}=w+\vartheta\id$.
\end{proposition}

\begin{proof}
By Theorem \ref{thm:werner-gram} it suffices to check that $\mu_n(r)$ is entangled. 
Firstly,
$$\tr\big(\mu_n(r)\eta_n(\widetilde{w})\big)
=\tau(r\widetilde{w})
=\tau(rw)+\vartheta \tau(r)
=-1+\vartheta <0$$
by \eqref{eq:entchar2} and Lemma \ref{lem:CSn2Hilbert}(1).
On the other hand, since $\mathfrak{t}_w+\vartheta$ is the trace of a sum of hermitian squares in $\cA$ by \eqref{eq:entchar2}, it attains nonnegative values on all tuples of projections from any von Neumann algebra $\cF$ with a tracial state $\omega$.
Therefore
\begin{equation}\label{eq:west}
0\ \le\ \vartheta +
\inf_{\substack{
	(\cF,\omega) \\
	X\in \cF^n \\
	X_j=X_j^\dag=X_j^2
}}
\mathfrak{t}_w(X) 
 \le\ \vartheta+
\inf_{\substack{
		X\in L(\C^n)^n \\
		X_j=X_j^\dag=X_j^2
}}
\mathfrak{t}_w(X) 
 =\ \vartheta+
\inf_{\substack{
		X\in L(\C^n)^n \\
		X_j=X_j^\dag=X_j^2
}}
T_w(X)
\end{equation}
where the last equality holds by \eqref{eq:mat2vna}. 
Note that $T_{\widetilde{w}}(X)=T_w(X)+\vartheta\tr(X_1)\cdots\tr(X_n)$ for every $X\in L(\C^n)^n$, and $\tr(P)\ge 1$ for every nonzero projection $P\in L(\C^n)$. Therefore \eqref{eq:west} implies
$$0 \le \inf_{\substack{
		X\in L(\C^n)^n \\
		X_j=X_j^\dag=X_j^2
}}
T_{\widetilde{w}}(X)\,.$$
Since $T_{\widetilde{w}}$ is multilinear and every positive semidefinite operator is a conic combination of projections, we conclude that $T_{\widetilde{w}}$ is nonnegative on all tuples of positive semidefinite operators on $\C^n$. Thus $\eta_n(\widetilde{w})$ is an entanglement witness for $\mu_n(r)$ by Theorem \ref{thm:trace_poly_formulation}.
\end{proof}

Given a state $r\in \C S_n$, let us consider the following trace polynomial optimization problem:
\begin{equation}
\label{eq:tracepoly}\tag{TPOP}
\begin{aligned}
\vartheta^* = &\inf_{
	\substack{ 
		\varepsilon \in \R,\ 
		w \in \C S_n
	}
} &&  \varepsilon \\
&\text{subject to}     && w=w^\dag \\
&                      && \tau(rw)  =-1\\
&                      && \mathfrak{t}_w + \varepsilon \geq 0 \text{ on } \cA \,.\\
\end{aligned}
\end{equation} 
This gives rise to the following hierarchy of SDP relaxations for \ref{eq:tracepoly}, indexed by $\ell\ge \lceil\frac{n}{2}\rceil$:
\begin{equation}
\label{eq:sdp_tracepoly}\tag{SDP-TPOP}
\begin{aligned}
\vartheta_\ell^*= &\inf_{
	\substack{ 
		\varepsilon \in \R,\ 
		w \in \C S_n, \\ 
		G
	}
} &&  \varepsilon \\
&\text{subject to}     && w=w^\dag \\
&&& G \succeq0\\
&                      && \tau(rw)  =-1\\
&                      && \mathfrak{t}_w + \varepsilon = \uptau\left( W_\ell^\dag G W_\ell\right) \,.\\
\end{aligned}
\end{equation} 
As a consequence of Proposition \ref{prop:oneway} we have:
\begin{corollary}
If $\vartheta_\ell^*<1$ for some $\ell\in\N$, then $\mu_n(r)$ is an entangled state. 
\end{corollary}

\begin{remark}\label{rem:size2}\rm
 Fix $n\in\N$. Since $\cM$ is a subset of tracial words in $\cA$, a very crude lower bound on the length of the vector $W_\ell$ is
$$M_\ell = \sum_{i=1}^\ell n(n-1)^{i-1} = n\frac{(n-1)^\ell-1}{n-2},$$
so the number of variables in the $\ell$th SDP \ref{eq:sdp_tracepoly} is at least exponential in $\ell$,
$$1+n!+\frac{(M_\ell+1)M_\ell}{2} = O((n-1)^{2\ell})\,.$$
\end{remark}

\section{Comparison of hierarchies}

Some remarks on the two SDP hierachies are in order.

The trace polynomial optimization framework in Proposition~\ref{prop:oneway} shares analogies with both Theorems~\ref{thm:trace_poly_formulation} and \ref{thm:werner-gram}. Like the latter, Proposition~\ref{prop:oneway} gives a dimension-independent certificate of entanglement. On the other hand, the trace polynomial context is closer to Theorem~\ref{thm:trace_poly_formulation}, although Proposition~\ref{prop:oneway} employs a different parametrization of witnesses (as it appeals to von Neumann algebras and their tracial states which are necessarily unital), leading to a dimension-independent statement.

However, it is important to mention that Proposition~\ref{prop:oneway} is possibly weaker than Theorem~\ref{thm:werner-gram} in the sense that it is unclear whether it detects entanglement of every entangled Werner state. 
While a positive resolution of the Connes embedding conjecture would likely imply the converse of Proposition~\ref{prop:oneway}, the former turned out to be false \cite{ji2020mipre}. 

Nevertheless, Proposition \ref{prop:oneway} leads to the hierarchy \ref{eq:sdp_tracepoly} for entanglement detection with smaller initial SDPs than 
the ones in \ref{eq:sdp_f_o_m}.
Comparing the number of variables from Remark~\ref{rem:size1} and ~\ref{rem:size2} we see the following:
for large $\ell$, the (commutative) \ref{eq:sdp_f_o_m} is much smaller than the (noncommutative) \ref{eq:sdp_tracepoly}. However, when utilizing SDP hierarchies in practice, one usually computes only the first few steps of the hierarchy, with the hope that they already give the sought answer. 
Since projections and tracial states of their products satisfy several relations, the first few steps of the second hierarchy \ref{eq:sdp_tracepoly} are actually much smaller than the first few steps of the first hierarchy \ref{eq:sdp_f_o_m}. Table \ref{tab:comparison} below compares the sizes of semidefinite constraints and numbers of equations in the first two steps of hierarchies ($\ell=\lceil \frac{n}{2}\rceil$ and $\ell=\lceil \frac{n}{2}\rceil+1$).

A further reduction is possible if one is interested in real states and real separability. Then one can take a coarser equivalence relation on $\cM$ that identifies $v$ and $v^\dag$ (thus $\uptau$ simulates a tracial state on a product of real projections) and restrict the scalars of $\cA$ to be real numbers. Encoding these additional symbolic constraints into $\cA$ decreases the number of tracial words of a given length, and thus decreases the size of the semidefinite constraint in the resulting analog of~\ref{eq:sdp_tracepoly}.

\begin{table}[!ht]
	\begin{tabular}{@{}c c l l l l l@{}}
	 &&\multicolumn{2}{l}  {\ref{eq:sdp_f_o_m}}
         &&\multicolumn{2}{l}  {\ref{eq:sdp_tracepoly}}\\
	$n$ && step 1 & step 2 && step 1 & step 2 \\
        \midrule
	$3$ && $(84, 211)$ & $(252, 925)$ && $(31,86)$  & $(109,443)$ \\
	$4$ && $(364, 1821)$ & $(1820, 18565)$ && $(53,246)$ & $(253,2432)$ \\
	$5$ && $(8855, 230231)$ & $(53130, 3108106)$  && $(491,9722)$ & $(2681,157492)$ \\
	\\
	\end{tabular}
	\caption{Pairs of sizes of semidefinite constraints and numbers of equations in \ref{eq:sdp_f_o_m} and \ref{eq:sdp_tracepoly} 
	for the first two steps in the hierarchies.
	\label{tab:comparison}
	}
\end{table}

\section{An example}\label{sec:exa}

In this section we use the second hierarchy  \ref{eq:sdp_tracepoly}
to detect entanglement in a four-qubit Werner state which has positive partial transposes across all bipartitions.
Let 
$s=41 \cdot \id+5\cdot(12)+5\cdot(34)+20\cdot(1234) \in \C S_4\,.$
There is a unique $r\in J_2\subset \C S_4$ such that\looseness=-1
$$\varrho = \mu_2(r)=\frac{\eta_2(ss^\dag)}{\tr(\eta_2(ss^\dag))}$$
is a four-qubit Werner state.
More explicitly, as in Remark \ref{rem:para} we get
\begin{align}\label{eq:4qb_ex}
r\,=\
&\tfrac{1}{24}\id
+\tfrac{1069}{34302}[(12) + (34)]
+\tfrac{7247}{274416}[(14) + (23)]
+\tfrac{6947}{274416}(13)
+\tfrac{7547}{274416}(24) \\
+&\tfrac{707}{34302}[(123) + (132) + (134) + (143)]
+\tfrac{1489}{68604}[(234) + (243) + (124) + (142)]\nonumber\\
+&\tfrac{8101}{548832}[(1324)+ (1423)]
+\tfrac{8251}{548832} [(1243) + (1342)]
+\tfrac{13171}{548832}[(1234) + (1432)]\nonumber\\
+&\tfrac{3811}{274416}(13)(24)
+\tfrac{6271}{274416}(14)(23)
+\tfrac{7651}{274416}(12)(34)\,.\nonumber
\end{align}
One can check that the partial transposes of $\varrho = \mu_2(r)$ are positive semidefinite for all bipartitions.
Consequently the Peres-Horodecki or PPT criterion does not detect entanglement in $\varrho$. 
However, already the first step ($\ell=\lceil\frac{4}{2}\rceil=2$) of the hierarchy \ref{eq:sdp_tracepoly} confirms that $\varrho$ is entangled. 
Since $r\in\R S_4$, it suffices to optimize over $w\in\R S_4$ and real symmetric $G$ in \ref{eq:sdp_tracepoly}. 
The numerical solution is $\vartheta_2 \approx 0.8537<1$,
from which a corresponding witness numerical $\widetilde{w}\in\R S_4$ as in Proposition \ref{prop:oneway} can be extracted.

Since $0.8537$ is close to 1, one might wish for an exact $w\in \mathbb{Q} S_4$ to clear doubts about numerical errors. 
To achieve this, we choose some rational $\vartheta_2'\in(\vartheta_2,1)$, for example $\vartheta_2'=\frac{9}{10}$, and solve the feasibility SDP
\begin{equation}
\label{e:wsw6}
w=w^\dag\,,
\quad G\succeq0\,,
\quad \tau(rw) = -1\,,
\quad \mathfrak{t}_w +\vartheta_2' =\uptau\left( W_\ell^\dag G W_\ell\right)\,.
\end{equation}

Geometrically, \eqref{e:wsw6} looks for a point in the intersection
of the positive semidefinite cone with an affine subspace.
In our example, the $53\times 53$ floating point solution $G$ produced by the 
interior-point method
SDP solver is positive definite. Therefore rationalizing, i.e., choosing
a sufficiently fine rational approximation of $G$, and then projecting onto the affine subspace will result in a rational solution of \eqref{e:wsw6}, cf.~\cite{PaPe08,CKP15}. 

Concretely, we obtain the exact dimension-free witness $\widetilde{w}=\frac{9}{10}\id+w\in\mathbb{Q}S_4$,

 \begin{align*}
\widetilde{w}\,=\ 
 	&\tfrac{70530553080581117}{73043335638912450} \id \\
 	 +&\tfrac{2153437054}{34127477475} [(12)+(34)]
 	 -\tfrac{1084798063661}{17296968712275} [(14) + (23)]
 	 -\tfrac{6399721673153}{58543548235200} (13)
 	 -\tfrac{166092679}{1576051425} (24)  \\
 	 -&\tfrac{128169}{202825} (12)(34)
 	 -\tfrac{112106999}{38636465420} (13)(24)
 	 -\tfrac{5}{66} (14)(23) \\
 	 +&\tfrac{441051017}{1988704319} [(234) + (243) + (124) + (142)]\\
 	 +&\tfrac{626723}{2766720}   [(123) + (132) + (134) + (143)]\\
 	 +&\tfrac{446599}{678600} [(1243) + (1342)]
 	 +\tfrac{23599}{171600} [(1324) + (1423)]
 	 -\tfrac{5220239}{3065280}  [(1234) + (1432)] \,.
 \end{align*}
The symmetry with respect to the parametrization of $r$ in ~\eqref{e:wsw6} is evident.

Note that due to Corollary~\ref{c:indep}, 
the state in Eq.~\eqref{eq:4qb_ex} is entangled in every dimension $d\ge2$.

\section{Additional remarks}
\def\cB{\mathcal{B}}

In this section we indicate how the techniques developed in this paper can be applied to non-Werner states and immanants.

\subsection{States invariant under a different unitary action}
It is well known that $n$-partite Werner states require fewer parameters (that is, $n!$) 
for their description than arbitrary $n$-partite states on $\cdn$ for $d>n$.
In this article we made use of this parametrization to remove the local dimension 
from the problem of detecting entanglement entirely.
This leads to the question: 
for which other sets of states can entanglement be detected in a dimension-free manner?

We presented our results for Werner states, 
however it is not hard to see that they can also be applied to 
quantum states $\varrho\in L(\cdn)$ that are invariant with respect to 
$U^{\otimes (n-k)} \ot \overbar{U}^{\otimes k}$ 
for any $k$. 
Such states are relevant for efficient port-based teleportation schemes~\cite{studzinski2021efficient}
and are elements of the walled Brauer algebra~\cite{Mozrzymas_2018}.
Thus they can be expanded in terms of partially transposed permutation operators, 
\begin{equation*}
\sum_{\sigma \in S_n} a_\sigma \eta_d(\sigma)^{T_k}\,,\qquad a_\sigma\in\C
\end{equation*} 
where $\cdot^{T_k}$ 
is the partial transpose acting on the last $k$ systems~\cite[Lemma 6]{PhysRevA.63.042111}.
As in the case of Werner states, it suffices to consider entanglement witnesses $\mathcal{W}$ for which the same invariance holds.

In contrast with $\eta_d$, the map $\widetilde{\eta}_d=\eta_d^{T_k}$ is not a $*$-representation of the algebra $\C S_n$. However, one can choose a ring structure on the vector space $\C S_n$ in a natural way, resulting in the aforementioned walled Brauer algebra $\cB_n$, so that the map $\widetilde{\eta}_d$ is a $*$-representation of $\cB_n$. By looking at the irreducible representations of $\cB_n$, one obtains a map $\widetilde{\mu}_d:\cB_n\to L(\cdn)$ by mimicking the construction of $\mu_d$ before, only now relying on a different ring structure (centrally primitive idempotents in $\cB_n$).
If $\varrho=\widetilde{\mu}_d(r)$ and $\mathcal{W}=\widetilde{\eta}_d(w)$ for some $r,w\in\cB_n$, then
$\tr(\mathcal{W} \varrho)$ equals the trace of $rw$ under the regular representation of $\cB_n$.
Similarly, the minimization of an operator containing partial transposes $\sum_{\sigma \in S_n} w_\sigma \eta_d(\sigma)^{T_k}$ over the set of separable states,
\begin{align*}
 &
 \min_{\ket{v_1}, \dots, \ket{v_n} \in \C^n} 
 \tr\big(\sum_{\sigma \in S_n} w_\sigma \eta_d(\sigma)^{T_k} \dyad{v_1} \ot \cdots \ot \dyad{v_k} \ot \cdots \ot \dyad{v_n})\big) \\
 = &
 \min_{\ket{v_1}, \dots, \ket{v_n} \in \C^n} 
 \tr\big(\sum_{\sigma \in S_n} w_\sigma \eta_d(\sigma) \dyad{v_1} \ot \cdots \ot \dyad{v_k}^T \ot \cdots \ot \dyad{v_n}^T)\big)\\
  = &
 \min_{\ket{v_1}, \dots, \ket{v_n} \in \C^n} 
 \tr\big(\sum_{\sigma \in S_n} w_\sigma \eta_d(\sigma) \dyad{v_1} \ot \cdots \ot \dyad{v_k} \ot \cdots \ot \dyad{v_n})\big)\,,
\end{align*}
reduces to that of an operator $\sum_{\sigma \in S_n} w_\sigma \eta_d(\sigma)$ with all partial transposes removed. Therefore nonnegativity of $\mathcal{W}=\widetilde{\eta}_d(w)$ on separable states corresponds to nonnegativity of $f_w$ on the spectrahedron $\cZ$ as before.
It follows that:
\begin{corollary}
Analogs of Theorems~\ref{thm:werner-gram} and \ref{thm:trace_poly_formulation}, Corollaries \ref{c:indep} and \ref{cor:dim-free}, 
 and the two hierarchies \ref{eq:sdp_tracepoly} and \ref{eq:sdp_f_o_m}
hold for states with 
$U^{\otimes (n-k)} \ot \overbar{U}^{\otimes k}$-invariance.
\end{corollary}

\subsection{Witnesses for arbitrary states}

Our approach also allows to detect entanglement in arbitrary states: given some state $\varrho \in L(\cdn)$, the twirl  
\begin{equation}\label{eq:twirl}
   E(\varrho) = \int_{U \in \mathcal{U}_d} U^{\otimes n} \varrho (U^\dag)^{\otimes n} dU
\end{equation}
yields a Werner state which can then be subjected to our hierarchies. 
Note that not every entangled state remains entangled under the twirling~\eqref{eq:twirl}.
The computation of the integral ~\eqref{eq:twirl} can be done in the following way~\cite{CollinsSniady2006, procesi2020note}.
Define
\begin{equation*}
   \Phi(\varrho) = \sum_{\sigma \in S_n} \tr(\sigma^{-1} \varrho) \eta_d(\sigma)
\end{equation*}
If $d\geq n$ then
\begin{equation*}
   E(\varrho) = \Phi(\varrho) \operatorname{Wg}(d,n)\,.
\end{equation*}
where $\Wg$ is the (Formanek-) Weingarten operator from Eq.~\eqref{eq:def_Wg}.
This yields an invariant state expanded in terms of the permutation operators, 
which can be subjected to our hierarchies \ref{eq:sdp_f_o_m} and \ref{eq:sdp_tracepoly}.

\subsection{Immanant inequalities}
We end with noting that the methods presented are directly applicable to the positivity of generalized matrix functions
[cf. Eq.~\eqref{eq:genmatfun}] and are of particular interest in the context of 
long-standing open conjectures on immanant inequalities~\cite{grone1988, Zhang2016, huber2021matrix}. 
For this is will likely be useful to take into account further symmetries~\cite{riener2013exploiting} and sparsity~\cite{klep2021sparse,nctssos}
in the semidefinite programs.

\bibliographystyle{alpha}
\bibliography{current_bib}
\end{document}